\documentclass[a4paper,11pt]{article}

\usepackage[utf8]{inputenc}
\usepackage[T1]{fontenc}
\usepackage{microtype}

\usepackage[margin=2.5cm]{geometry}

\usepackage{amsmath,mathtools}
\usepackage{amsthm,thmtools}
\usepackage{amssymb}
\usepackage{bm,dsfont}
\usepackage{csquotes}
\usepackage{xcolor}
\usepackage{xspace}

\usepackage{paralist}

\usepackage[small]{caption}
\usepackage{subcaption}
\usepackage{tikz}
\usepackage{pgfplots}
\pgfplotsset{compat=newest}
\usepackage{booktabs,tabularew}
\usepackage[para]{threeparttable}
\usepackage[ruled,linesnumbered,noline]{algorithm2e}
\usepackage{array}

\usepackage[skip=5pt plus 3pt minus 1pt,parfill=0pt]{parskip}
\makeatletter
\def\thm@space@setup{%
  \thm@preskip=\parskip \thm@postskip=0pt
}
\makeatother

\usepackage[
 backend=biber,
 style=alphabetic,
 citetracker,
 hyperref=auto,
 maxcitenames=5,
 sortcites,
 sorting=ynt,
 maxbibnames=12,
 date=year,
 isbn=false,
 url=false,
 doi=false,
 eprint=false,
]{biblatex}

\newbibmacro{string+doiurlisbn}[1]{
  \iffieldundef{doi}{
    \iffieldundef{url}{
      #1
    }{
      \href{\thefield{url}}{#1}
    }
  }{
    \href{http://dx.doi.org/\thefield{doi}}{#1}
  }
}

\DeclareFieldFormat
[article,inbook,incollection,inproceedings,patent,thesis,unpublished]
  {title}{\usebibmacro{string+doiurlisbn}{#1}}

\usepackage[colorlinks]{hyperref}
\hypersetup{breaklinks=true,citecolor=green!60!black,linkcolor=blue!70!black,urlcolor=blue!70!black}
\usepackage[capitalize,nameinlink]{cleveref}
\usepackage{xurl}

\declaretheorem{theorem}
\declaretheorem{lemma}

\declaretheorem{corollary}

\DeclarePairedDelimiter\ceil{\lceil}{\rceil}

\newcommand{\bone}{\mathbf{1}}
\newcommand{\bigo}{\mathcal{O}}
\renewcommand{\ceil}[1]{\left\lceil #1 \right\rceil}

\newcommand{\alg}{{\mathrm{ALG}}}
\newcommand{\opt}{{\mathrm{OPT}}}

\newcommand{\cP}{\mathcal{P}}

\newcommand{\PS}{\mathrm{PS}}
\newcommand{\PR}{\mathrm{PR}}
\newcommand{\DS}{\mathrm{DS}}
\newcommand{\DR}{\mathrm{DR}}

\newcommand{\virt}{E_{\mathrm{virt}}}
\newcommand{\phys}{E_{\mathrm{phys}}}

\title{Indirect Coflow Scheduling\thanks{This research was conducted while the authors were participants in the Simons Institute program on Algorithmic Foundations for Emerging Computing Technologies. }}
\date{}
\author{Alexander Lindermayr\thanks{Institut für Mathematik, Technische Universität Berlin, Germany.} \and Kirk Pruhs\thanks{Computer Science Department, University of Pittsburgh, USA. Supported by National Science Foundation grant CCF-2209654.}  \and Andréa W.~Richa\thanks{School of Computing and Augmented Intelligence, Arizona State University, USA. Supported in part by research grants NSF-CCF-2312537 and 2106917;
and U.S.\ ARO (MURI W911NF-19-1-0233).} \and Tegan Wilson\thanks{Khoury College of Computer Sciences, Northeastern University, Boston MA, USA.}}

\begin{document}
\maketitle

\begin{abstract}
We consider routing in reconfigurable networks, 
which is also known
as coflow scheduling in the literature.
The algorithmic literature  generally (perhaps implicitly) 
assumes that the amount of data to be transferred is large. 
Thus the standard way to model a collection of requested data transfers is by an integer demand matrix~$D$,
where the entry in row $i$ and column $j$ of $D$ is an integer representing the amount of information that the application wants 
to send from machine/node~$i$ to machine/node~$j$. 
A feasible coflow schedule is then a sequence of matchings,
which represent the sequence of data transfers that covers $D$.
In this work, we investigate coflow scheduling when the size of some of
the requested data transfers may be small relative to the amount
of data that can be transferred in one round.
In particular, we investigate algorithms that employ
fractional matchings and/or that employ indirect
routing, and compare the relative utility of these options. 
We design algorithms that perform much better for
small demands than the algorithms in the literature that were designed
for large data transfers. 
%
%\keywords{coflow scheduling, reconfigurable networks, indirect routing, online, approximation, completion time, makespan}
\end{abstract}

\thispagestyle{empty}

\newpage

\section{Introduction}

We consider routing in reconfigurable networks, also
known as {\em  coflow scheduling}, whose topologies dynamically adapt to the structure of the communication traffic~\cite{StoicaCoflow,CACM}.
There is 
currently no consensus in the community on 
what the best reconfigurable network design is,
and the community lacks models and metrics to rigorously
study and compare different designs~\cite{CACM}.
Our high level goal in this work is to fill some of this gap.

\subsection{Background and Motivation}

In coflow scheduling, there is a layer, which 
we will call the coflow layer, 
that sits  underneath the application layer, and  above the network layer, in the network protocol stack.
The application passes the coflow protocol a collection of 
requests for data transfers that the application wants to execute.
The task of the coflow layer is to sensibly schedule/manage 
when these requests for data transfers are passed to the
networking layer. 
Most notably the coflow layer should not overwhelm the network with more requests 
than the network can simultaneously support. 

The algorithmic literature in this area generally 
models this informal goal by the formal constraint
that at all times the pairs of nodes communicating form
a matching $M$, which in this context is a set %collection
$M =  \{ (s_1, t_1), (s_2, t_2), \ldots, (s_n, t_n)\}$ 
of source-destination pairs of nodes with distinct sources
and distinct sinks (so $i \ne j $ implies $s_i \ne s_j$ and $t_i \ne t_j$).
A coflow schedule is then a 
sequence of matchings $M_1, \ldots, M_T$, where matching
$M_t$ is passed to the networking layer at time~$t$,
and during the time interval $[t, t+1]$ the network layer 
transmits a unit of data from~$s_i$ to~$t_i$ for
each $(s_i, t_i) \in M_t$.

The input passed to the coflow layer is assumed to consist of 
 an integer demand matrix $D$,
where $D_{ij}$, the entry in row $i$ and column~$j$ of $D$, is an integer representing the amount of information that the application wants 
to send from machine/node $i$ to machine/node~$j$ (see, e.g.,~\cite{ChowdhuryS15,ChowdhuryZS14,AhmadiKPY20,ChowdhuryKPYY19,ImMPP19,RohwedderS25,JahanjouKR17,QiuSZ15,DinitzM20}). 
Then the objective is to optimize some 
standard quality-of-service metric, like minimizing {\em makespan},
which is the time when the last unit of data is transmitted, or minimizing the {\em average completion (or delay) time}, which is the average time that it takes a unit of data to 
reach its destination.
Some of the coflow scheduling literature also considers generalizations of
this basic setting, see \Cref{subsec:relatedwork}.

Originally, the proposal of a coflow layer
was to handle large numbers of data transfer requests,
whose sizes were also large~\cite{StoicaCoflow},
and so this has been the, perhaps implicit, assumption
in most of the algorithmic literature.
This is reflected in the modeling assumption that
the demands $D_{ij}$ are integer. 
Our goal in this work was to examine coflow 
scheduling when 
the amount of data transfer requests may be small,
as one could conceive in certain data-parallel applications.
We model this by allowing the demands $D_{ij}$
to be arbitrary nonnegative rational numbers;  
in particular demands may be less than 1.

To understand the additional issues that arise if some
of the requested data transfers are small, consider
the demand matrix $D$ where each entry $D_{ij}=\frac{1}{2n^{3/2}}$, where $n$ is the 
number of nodes used by the application. 
Then these data transfers can be supported by routing 
up to a unit of data from node $i$ to node $j$ at the time $t$ where 
$t = (j-i) \mod n$. This results in a makespan of  $n-1$ (and average completion
time of $(n-1)/2$),
and it might seem that obtaining sublinear makespan  is impossible,  
as $D$ cannot be expressed as a linear combination of
less than $n-1$ matchings. 

However, a makespan of $O(\sqrt{n})$ can 
be achieved  if the coflow
scheduler can indirectly route data, as it is commonly assumed
in the reconfigurable networks 
literature~\cite{shoal,sirius,shale,rotornet,rotornet-2,mars}. 
As an example of indirect routing,
a unit of data with source $s_i$ and destination $t_i$ could be routed using two
matchings $M_1$ and $M_2$: $M_1$ could send the data from source node $s_i$ to
an intermediate node $k$, and $M_2$ could route this 
data from node $k$ to destination node $t_i$.
One can achieve a makespan of $O(\sqrt{n})$ (for the instance where each $D_{ij}=\frac{1}{2n^{3/2}}$) in the following way.
The $n$ nodes are conceptually assigned
to nodes in an $\sqrt{n}\times \sqrt{n}$ square grid $G$.
Then the data transfers are divided into two phases,
and each phase is divided into $\sqrt{n}- 1 $ subphases.
In subphase $k$ of Phase 1, the node in row $i$ and column $j$
of $G$ routes to the node in row $[(i + k) \mod \sqrt{n}]$ and column $j$
all of its data whose final destination is in row $[(i + k) \mod \sqrt{n}]$. Thus after the first phase, the current location of
each unit of data is in the same row as its final destination.
Then subphase $k$ of Phase 2, the node in row~$i$ and column~$j$
of $G$ routes to the node in row $i$ and column $[(j + k) \mod \sqrt{n}]$
all of its data whose final destination is that node. 
See \Cref{fig:3x3-grid} for an illustration of this algorithm 
when $n=9$. 
\begin{figure}[tb]
    \centering
    \includegraphics[width=0.42\textwidth]{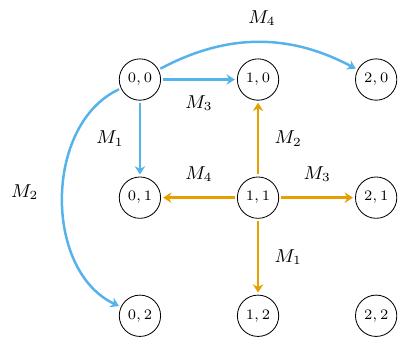}
    \caption{The data transmissions for nodes $(0, 0)$ and $(1, 1)$ for the $3\times 3$ grid formed from $n=9$ nodes. The matchings $M_1$ and $M_2$ belong to Phase 1, and the matchings
    $M_3$ and $M_4$ belong to Phase 2.}
    \label{fig:3x3-grid}
\end{figure}

Here we consider eight variations of coflow scheduling
with rational demand matrices, where ``eight'' comes from there
being two options to each of the following three characteristics:
\begin{description}
\item[Direct  vs. Indirect Routing.] We consider both allowing and disallowing 
indirect transfers.
    \item[Makespan vs. Average Completion Time Objective.] We consider both of the two most commonly considered quality-of-service objectives,
makespan and average completion time (which are essentially the infinity norm and one norm of the delays for the data units). 
\item[Fractional  vs. Integral Routing.] 
One conceptual ``solution'' for coflow scheduling of 
small data transfers is to allow
the network layer to
route a fractional matching in each unit of time. 
In this context, a fractional matching is given by a collection
$$ \{ (s_1, t_1, p_1), (s_2, t_2, p_2), \ldots, (s_n, t_n, p_n)\} $$
of triples
$(s_i, t_i, p_i)$ indicating that $p_i$ fractional units of
data are routed from $s_i$ to $t_i$, such that
for all nodes $v$, 
$\sum_{i : s_i=v} p_i \le 1$
and $\sum_{i : t_i=v} p_i \le 1$.
Note that a fractional matching transfers the same amount of data per unit time as an integer matching, so
it does not violate the unit throughput constraints at the network layer and thus it would be feasible in several networking scenarios. In contrast, such a relaxation is probably not feasible in 
settings where the 
network is optical
and the matching is being implemented by the rotation of 
mirrors~\cite{CACM,lightwave-fabrics,rotornet-2}.
Thus we consider both allowing fractional matchings and restricting
matchings to be integral. 
\end{description}

For each
of these eight variations, the goal is to find the best approximation ratio
that is achievable by a polynomial-time algorithm.
\Cref{tab:results} summarizes the best-known approximation ratios for all the variations of coflow scheduling.

\setlength{\tabcolsep}{10pt}
\begin{table}[h]
\centering
\caption{Overview of known approximation guarantees. We have provided references to some of the results, as appropriate; for the others, the corresponding approximation bounds follow directly from well-known standard approximation techniques.}\label{tab:results}
\begin{tabular}{lllll}
\toprule
 & \multicolumn{2}{l}{fractional matching} \quad & \multicolumn{2}{l}{integral matching} \smallskip \\           
& direct     & indirect       & direct     & indirect     \\ \midrule
makespan     &1~\cite{schrijver}   & 1~\cite{schrijver}  & 1    & $O(\log n)$     \\ \midrule
average completion & 1   &  $O(\log n)$  & $\sqrt{2}$~\cite{GandhiM09} &  $O(\log n)$    \\ 
\bottomrule
\end{tabular}
\end{table}

For direct integral routing for the makespan objective, optimal
schedules can be computed using an optimal algorithm
for edge coloring a bipartite graph~\cite{schrijver}.
For fractional routing with the makespan objective, 
the optimal makespan is the same for both direct and
indirect routing because the fractional edge coloring number
for a bipartite graph is the same as the maximum 
coloring number of the two underlying bipartite graphs~\cite{schrijver}.
Finally it is straightforward to model fractional direct
routing to optimize average completion time with a polynomial-size linear program. An $O(\log n)$-approximation is achievable 
for the rest of the variants by modeling the problem
as a linear program, where the edges in each time
step are over-capacitated by $O(\log n)$~\cite{Gonnet81}. It can be established
that this over-capacitation is sufficient via the probabilistic method.
For direct integral routing for average completion time, rounding the demands up to the next integer essentially does not change the problem, 
one can use the $\sqrt{2}$-approximation for integer
demands is also $\sqrt{2}$-approximate for rational demands~\cite{GandhiM09}.

In this paper, we focus on
the variants where a $O(1)$-approximation does not (to the best of our knowledge)
follow directly from known results,
which all involve indirect routing. 

\subsection{Summary of Our Results}
We present  two main results, stated in Theorems~\ref{thm:main1} and \ref{thm:main2}. Theorem~\ref{thm:main1} presents the {\em first $O(1)$-approximation algorithm for indirect fractional routing for the average completion time} objective:

\begin{theorem}
\label{thm:main1}
There is a polynomial-time algorithm $\alg$ for indirect routing with fractional matchings that outputs, for any valid instance, a schedule whose average completion time is at most 16 times the optimal average completion time. The algorithm only uses direct routing schedules, so the approximation bound is valid for both the direct and indirect routing settings with fractional matchings.
\end{theorem}

The algorithm $\alg$ is the natural greedy algorithm that
routes an arbitrary maximal matching at each unit of time. 
Perhaps the limited utility of indirectness in the
context of fractional routing for the objective of 
average completion time is somewhat surprising.
Although recall that a similar situation arises for
the makespan objective, where 
indirect routing is of no benefit.

Our algorithm 
introduces a novel dual fitting framework. As is usually the case in approximation algorithms,
the critical step is to find the ``right'' lower bounds
to the optimum to compare against. 
To that end,
we introduce two relaxations
of the problem. In the {\em Sender  Bound Problem}, there are only
capacity constraints on the amount of data a
sender can send in a unit of time, but there are no constraints
on the amount of data a receiver can receive in a unit of time.
Similarly, in the {\em Receiver Bound Problem}, 
there are only
capacity constraints on the amount of data a
receiver can receive in a unit of time, but there are no constraints
on the amount of data a sender can send in a unit of time.
We consider the natural linear programming formulations
PS and PR for the sender and receiver bound problems,
and their respective duals DS and DR. 
We then use dual feasible solutions for DS and DR as our lower bounds. 
We are not aware of a similar ``double'' dual fitting analysis in
the literature.

Our second main result appears in \Cref{thm:main2},  
with a full characterization of
the {\em worst-case makespan and average completion time} for {\em indirect integral routing},  
with respect to two parameters $n$ and $B$. 
In the future, we plan to leverage the bounds in Theorem~\ref{thm:main2} towards determining whether a $O(1)$-approximation for indirect integral routing is possible.

\begin{theorem}\label{thm:main2}
 Let $D$ be a rational demand
matrix where the maximum row or column sum is at most $B$.
There is a polynomial-time algorithm $\alg$
that uses indirect integral routing, and  guarantees makespan and average completion time
\[
\begin{cases}
    O(\log n) &\quad \text{ if } B \leq 2 \ , \\
    O(\frac{B \log n }{ \log B}) &\quad \text{ if } 2 \leq B \leq n\ , \text{ and} \\
    O(B) &\quad \text{ if } n \leq B\ . \\
\end{cases}
\]
Further if
each entry of $D$ is $\frac{B}{n}$ then the optimal
schedule has makespan and average completion time 
$\Omega(\max\{\log n, \frac{B \log n }{ \log B}, B \})$. Thus, algorithm $\alg$ is worst-case optimal with respect to both objectives.
\end{theorem}

Theorem~\ref{thm:main2} confirms the intuition that the 
worst-case input should be when the demand is evenly
spread out. As $\ceil{B}$ is a clear lower bound on the optimal makespan (and 
optimal average completion time), we also derive the following corollary to the theorem:
\begin{corollary}
 Algorithm $\alg$ from \Cref{thm:main2} attains an $O(\log n)$-approximation ratio for both makespan and average completion times, for all $B$. 
\end{corollary}
In some sense, the challenge of obtaining an $O(1)$-approximation 
appears to be to identify instances that admit schedules with less than worst-case objective value.

If $B \leq 2$, we can use Valiant's hypercube design~\cite{Valiant82} to achieve $O(\log n)$ makespan by using routing paths with $O(\log n)$ intermediate nodes. 
In a nutshell,
this approach  tries to maximize
the usage of intermediate nodes, because there is very little demand to route. 
Accordingly, the $\Omega(\log n)$ makespan lower bound
relies purely on topological properties, and does 
not take demands into account.
Alternatively, if $B \geq n$, we can use 
a single round-robin design and route flow directly from 
sources to destinations in maximum matchings to achieve $O(B)$  
makespan---there is no need for intermediate nodes because there are no small demands that can be ``combined.''
The $\Omega(B)$ lower bound relies purely on the maximum demand requested per source or destination node, without taking the
network topology into account.

The challenge in the $2 \leq B \leq n$ regime is thus to interpolate between these two extreme cases.
For the upper bound, we use a generalization of the hypercube design, which uses $dn^{1/d}$ distinct matchings, for some dimension parameter $d$.
Each matching must be repeated with a multiplicity $m$ that depends on $B,n,$ and $d$, and increases as $d$ increases, resulting in $mdn^{1/d}$ total makespan.
Setting $d$ appropriately results in the desired makespan bound $O(B \log n / \log B)$.
To prove the matching lower bound, we carefully partition routings into two categories. 
If a source can route to a large fraction of destinations using ``short'' paths with few intermediate nodes (even if only a small fraction of demand is routed there), we show a lower bound purely on network topology properties. 
Otherwise, if the routing uses ``long'' paths with many intermediate nodes for a large fraction of the total demand, we can prove a large makespan via a topology-independent averaging argument.
This results in the desired lower bound $\Omega(B \log n / \log B)$ for an appropriate threshold of ``short'' versus ``long'' paths.

\subsection{Related Work}
\label{subsec:relatedwork}

Coflow scheduling was popularized by
Chowdhury and Stoica~\cite{StoicaCoflow}, and followed by 
heuristics for online and offline settings~\cite{ChowdhuryZS14,ChowdhuryS15}.
A recent overview of reconfigurable networking can
be found in \cite{CACM}.

There is some literature on direct integral routings to 
optimize weighted completion time. 
Marx~\cite{Marx09} showed that this problem is APX-hard.
The natural greedy algorithm is a 2-approximation~\cite{Bar-NoyBHST98,Bar-NoyHKSS00}. 
Gandhi and Mestre~\cite{GandhiM09} give an improved $\sqrt{2}$-approximation, and
Halldórsson, Kim, and Sviridenko~\cite{HalldorssonKS11} show a $1.83$-approximation for a more general setting.
Other works in this area consider non-preemptive variants, where large demands must be scheduled consecutively in time~\cite{HalldorssonKS01,Kim05,GandhiHKS08}.
The fractional preemptive variants of these problems are also captured by the general polytope scheduling problem, as any feasible instantaneous allocation belongs to the matching polytope~\cite{ImKM18,JLM25}.

There is also some literature on the situation where
there are several different applications, each 
with their own demand matrix, and the objective is to
minimize the sum of the makespans over the different applications.
This literature considers direct integral routing for the average
completion time objective. 
The problem is APX-hard~\cite{Marx09}, 
and an even stronger hardness-of-approximation result ruling out any factor better than $2$ follows
from a reduction to the concurrent open shop problem~\cite{QueyranneS02}, which is hard under the Unique Games Conjecture~\cite{BansalK09}.
The first constant-factor approximation was shown by Qiu, Stein, and Zhong~\cite{QiuSZ15}, and after several improvements
\cite{KhullerP16,ShafieeG17,AhmadiKPY20}, the current
best approximation ratio is $3.415$~\cite{RohwedderS25}.
The special case where each coflow corresponds to the set of incident 
edges of a vertex is called \emph{data migration} and also 
has been studied extensively~\cite{Kim05,KhullerKW04,GandhiHKS08,GandhiM09,Mestre10}.
Other notable results are
online algorithms~\cite{KhullerLSSV19,BhimarajuNV20,DinitzM20,LLM25},
a matroid generalization~\cite{ImMPP19},
and generalization for general graphs~\cite{ChowdhuryKPYY19,JahanjouKR17}.
In particular some elements of our dual fitting analysis 
derive from the dual fitting analysis in \cite{DinitzM20}.

The second main area of related work is in Oblivious Reconfigurable Networks (ORNs).
The theoretical problem was formalized recently~\cite{AmirWSWKA22}, modeling the constraints of various proposed data center architectures using fast circuit switch technology~\cite{rotornet,rotornet-2,shoal,mars,shale}.
The ORN model allows for multi-hop indirect routing, its key connection to our work.
However, it also focuses on the streaming setting, where traffic demands arrive over time and must be continuously fulfilled.
Additionally, they focus on obliviousness, and frame their results as a guarantee on the minimum amount of traffic which can be fulfilled, rather than as a guaranteed approximation ratio.
One contribution of the ORN area is a generalized hypercube-style design with a tunable hop parameter~\cite{AmirWSWKA22,WilsonASWK23,shale}.
Our work in \Cref{sec:indirect-integral} uses similar ideas, but applied to a different use case.
Randomized oblivious connection schedules with some adaptive routing have also been considered in theory~\cite{WilsonASKSW24}, focusing on the same tradeoff guarantees as in the oblivious literature.
It is still unknown what can be done with a fully adaptive design.

\section{Notation and Formal Definitions}

We follow broadly the definitions and notations used in \cite{AmirWSWKA22}.
Our setting is a  network with a collection $N = \{0,\ldots,n-1\}$ of $n$ nodes.
The input is an $n \times n$ demand matrix $D$ with nonnegative rational entries, where $D_{ij}$ specifies the amount of data that needs to be transferred from node $i$ to node $j$.

We define 
a particular flow network
derived from the demand matrix $D$ that 
allows us to 
characterize 
feasible solutions 
for the problems we consider. 
The graph (or network) $G$ has
vertex set $N \times \{0,\ldots,T\}$, for some  $T\in \mathbb{N}$.
The edge set of $G$  is partitioned into edges $\phys$ that we will label
as \emph{physical}, and edges $\virt$ that we will label as \emph{virtual}.
The physical edges are of the form $(i,t) \to (j,t+1)$ for all $i \neq j \in N$ and $t \in \{0,\ldots,T-1\}$. 
The virtual edges are of the form  $(i,t) \to (i,t+1)$ for all $i \in N$ and $t \in \{0,\ldots,T-1\}$.
For all $i,j \in N$ and $t \in \{1,\ldots,T\}$,
let $\cP(i,j,t)$ be the set of all directed paths from vertex $(i,0)$
to vertex $(j,t)$ in~$G$. 
We define $\cP:= \cup_{i,j \in N, t \geq 1} \cP(i,j,t)$.
An \emph{indirect fractional routing}
is a flow $f : \cP \to \mathbb R_{\geq 0}$.
We say that the routing (or flow) $f$ is \emph{feasible} for demands $D$ 
if it satisfies the following constraints:
\begin{description}
    \item[Capacity Constraints:] each vertex has a total outgoing flow over physical edges of at most $1$, that is, for each $t \in \{0,\ldots,T-1\}$ and $i \in N$, 
    \[
    \sum_{j \in N \setminus \{i\}} \sum_{P \in \cP} f(P) \cdot \bone[((i,t), (j,t+1)) \in P] \leq 1 \ ,
    \] 
    and each vertex has a total incoming flow over physical edges of at most $1$, that is, for each $t \in \{0,\ldots,T-1\}$ and $j \in N$, 
    \[
    \sum_{i \in N \setminus \{j\}} \sum_{P \in \cP} f(P) \cdot \bone[((i,t), (j,t+1)) \in P] \leq 1 \ .
    \] 
    \item[Demand Satisfaction:] for each $i,j \in N$, $\sum_{t=1}^T \sum_{P \in \cP(i,j,t)} f(P) \geq D_{ij}$.
\end{description}
If a routing
$f$ routes $\alpha$ units of flow
along a path $P$ from $(i, 0)$ to $(j, t)$ that
contains a physical edge $(a, s) \rightarrow (b, s+1)$, 
then $\alpha$ units of data with source $i$ and destination 
$j$ are transmitted from node $a$ to node $b$ at time $s$.

A routing
$f$ is an \emph{integral} routing if every vertex in $G$ has
at most one incoming physical edge with positive flow,
and at most one outgoing physical edge with positive flow; hence, for all $t \in \{1,\ldots,T\}$, the physical edges between vertex sets $\{(i,t)\}_{i \in N}$ and $\{(j,t+1)\}_{j \in N}$ with positive flow form a matching in $G$.
A routing
$f$ is a \emph{direct} routing
if it sends positive flow only over paths that contain at most one physical edge.
A direct fractional routing and a direct integral routing can thus
be viewed as a sequence of fractional or integral matchings $M_0,\ldots,M_{T-1}$ in the complete bipartite
graph 
between two copies of the node set $N$ 
that covers~$D$, respectively.
Note that with our definitions, direct 
routings are a subset of indirect 
routings, and integral routings are a subset of fractional routings.

We say that a feasible routing $f$ for the network $G$ has a \emph{makespan} of $T$, if the last unit of data was delivered at time $T$.
The \emph{completion time} $t$ for a unit of data with source $i$ 
and destination $j$ is the smallest $t$ such that this data
arrives at a node of the form $(j, t)$. 
Thus, the total completion time of a schedule is then 
\[
 \sum_{t = 1}^T \sum_{i,j \in N} \sum_{P \in \cP(i,j,t)} t \cdot f(P) \ .
\]

\section{Fractional Matching, Indirect Routing, and Average Completion Time}

In this section we prove \Cref{thm:main1}.
For clarity of explanation, we prove the theorem  
for the total sum of completion times,
which is equivalent to the average completion time (up to a factor of $1/\sum_{i,j} D_{ij}$).
We start with a description of the polynomial-time algorithm $\alg$ 
that gives us \Cref{thm:main1}.

\paragraph{Algorithm $\bm{\alg}$:}
At each time $t$, the fractional matching $M_t$ is an arbitrary maximal fractional
matching of the residual demand matrix $D(t)$, where for all $i, j \in N$, 
\[
    D_{ij}(t) := D_{ij} - \sum_{t'=0}^{t-1} (M_{t'})_{ij} 
\]
where $(M_{t'})_{ij}$ denotes the amount of data on edge $(i,j)$ in $M_{t'}$.

\smallskip

The most important step in obtaining approximation ratio results is to find ``good'' lower bounds to the optimal objective value, which we will denote by $\opt$.
Here we consider the following two lower bounds, each of which are derived by relaxing the constraints in two different (but symmetric) ways.

\paragraph{Sender Bound Problem:}
For each fractional  matching
$M= \{ (s_i, r_i, p_i) \}$ 
and for all nodes $v$, we require that
$\sum_{i : s_i=v} p_i \le \frac{1}{4}$.
However, 
 $\sum_{i : r_i=v} p_i$ can be unbounded.

\paragraph{Receiver Bound Problem:}
For each fractional  matching
$M= \{ (s_i, r_i, p_i) \} $
and
for all nodes $v$, we require that
$\sum_{i : r_i=v} p_i \le \frac{1}{4}$.
However,
 $\sum_{i : s_i=v} p_i$ can be unbounded.

\medskip

In the sender bound problem there is no bound on the  rate that data can be transferred into a node,
and in the receiver bound problem there is no bound on the rate that data  can be transferred out of a node. 
Instead of limiting the amount of flow that can be transferred into or out of a node to one, as it may seem natural, we limit it to $\frac{1}{4}$.
This setting arises due to technical reasons and makes our dual fitting arguments more convenient.  
Note that 
\begin{align}
   \opt^S \le 4  \cdot \opt
 \text{ and }\opt^R \le 4 \cdot \opt
 \label{eq:optrs}
\end{align}
where $\opt^S$ and $\opt^R$ are optimal solutions  to the sender bound and receiver bound problems, respectively. This is because replicating each matching
in $\opt$ four times results in a solution that is feasible
for both the sender bound problem and the receiver bound
problem, and increases the objective value by at most a factor
of 4. 

We now give a linear programming
formulation $(\hyperlink{ps}{\PS})$ for the sender bound problem, 
and a linear programming formulation $(\hyperlink{pr}{\PR})$ for the receiver bound problem. For convenience we are going to rescale time
so that every demand $D_{ij}$ is an integer. Such
a rescaling does  not affect the approximation ratio.
In each formulation, the intended meaning of the variable $x_{ijt}$ is
the 
amount of data with source node $i$ destination node $j$ that is transmitted at time $t$. 

\smallskip

\noindent\fbox{%
    \begin{minipage}{\dimexpr\linewidth-2\fboxsep-2\fboxrule}

{
 \renewcommand{\arraystretch}{1.3}
 \[\begin{array}{rr>{\displaystyle}rc>{\displaystyle}l>{\;}l}
  \hypertarget{ps}{(\PS)} & \operatorname{min} &\multicolumn{3}{l}{\displaystyle \sum_{i \in N} \sum_{j \in N} \sum_{t=0}^{T} t \cdot x_{ijt} }   \\
  &\text{s.t.} &\sum_{t=0}^{T} x_{ijt} &\ge& D_{ij} & \forall i,j \in N   \\
   && \sum_{j \in N} x_{ijt} &\le& \frac{1}{4} &\forall i \in N,\  0 \leq t \leq T \\
  && x_{ijt}  &\ge& 0 &\forall i,j \in N, 0 \leq t \leq T
 \end{array}\]%
}%
\end{minipage}
}

\smallskip

\noindent\fbox{%
    \begin{minipage}{\dimexpr\linewidth-2\fboxsep-2\fboxrule}

{
 \renewcommand{\arraystretch}{1.3}
 \[\begin{array}{rr>{\displaystyle}rc>{\displaystyle}l>{\;}l}
  \hypertarget{pr}{(\PR)} & \operatorname{min} &\multicolumn{3}{l}{\displaystyle \sum_{i \in N} \sum_{j \in N} \sum_{t=0}^{T} t \cdot x_{ijt} }   \\
  &\text{s.t.} &\sum_{t=0}^{T} x_{ijt} &\ge& D_{ij} & \forall i,j \in N   \\
   && \sum_{i \in N} x_{ijt} &\le& \frac{1}{4} &\forall j \in N,\ 0 \leq t \leq T \\
  && x_{ijt}  &\ge& 0 &\forall i,j \in N, 0 \leq t \leq T
 \end{array}\]%
}%
\end{minipage}

}
\smallskip

We now give the dual linear program $(\hyperlink{ds}{\DS})$ of 
    the primal linear program $(\hyperlink{ps}{\PS})$, and the dual linear program
$(\hyperlink{dr}{\DR})$ of the primal linear program $(\hyperlink{pr}{\PR})$. 
In these dual programs the $\alpha$ variables are associated with the first set of constraints in the primal, and the $\beta$ variables are associated with the second set of constraints. 

\smallskip

\noindent\fbox{%
    \begin{minipage}{\dimexpr\linewidth-2\fboxsep-2\fboxrule}
{
 \renewcommand{\arraystretch}{1.9}
 \[\begin{array}{rr>{\displaystyle}rc>{\displaystyle}l>{\quad}l}
  \hypertarget{ds}{(\DS)}& \operatorname{max} &\multicolumn{3}{l}{\displaystyle\sum_{i,j \in N} D_{ij} \cdot \alpha^S_{ij} - \sum_{i \in N} \sum_{t=0}^T \beta^S_{it}} \\
  &\text{s.t.} & \alpha^S_{ij} - t  &\le& 4 \cdot \beta^S_{it}  &\forall i,j \in N,\ 0 \leq t \leq T \\
  &&\alpha^S_{ij}, \beta^S_{it} &\ge& 0 & \forall i,j \in N,\ 0 \leq t \leq T
 \end{array}\]
}%
\end{minipage}
}

\smallskip

\noindent\fbox{%
    \begin{minipage}{\dimexpr\linewidth-2\fboxsep-2\fboxrule}
{
 \renewcommand{\arraystretch}{1.9}
 \[\begin{array}{rr>{\displaystyle}rc>{\displaystyle}l>{\quad}l}
  \hypertarget{dr}{(\DR)} & \operatorname{max} &\multicolumn{3}{l}{\displaystyle\sum_{i,j \in N} D_{ij} \cdot \alpha^R_{ij} - \sum_{j \in N} \sum_{t=0}^T \beta^R_{jt}} \\
  &\text{s.t.} & \alpha^R_{ij} - t  &\le& 4 \cdot \beta^R_{jt}  &\forall i,j \in N,\ 0 \leq t \leq T \\
  &&\alpha^R_{ij}, \beta^R_{jt} &\ge& 0 &  \forall i,j \in N,\ 0 \leq t \leq T
 \end{array}\]
}%
\end{minipage}
}

\medskip

We now define dual feasible solutions for $(\hyperlink{ds}{\DS})$ and $(\hyperlink{dr}{\DR})$ related to the execution of the algorithm $\alg$. 
Recall that $D_{ij}(t)$ denotes the amount of data with source $i$ and destination $j$ that $\alg$ has not 
transmitted by time $t$. 
Let $D^S_i(t) = \sum_{j} D_{ij}(t)$ be the amount of data with source $i$ that $\alg$ has not 
transmitted by time $t$, and let $D^R_j(t) = \sum_{i} D_{ij}(t)$ be the amount of data with 
destination $j$ that $\alg$ has not transmitted by time $t$.
Finally let $D^S_i$ denote $D^S_i(0)$ and $D^R_j = D^R_j(0)$.

Now, we can define dual solutions for $(\hyperlink{ds}{\DS})$ and $(\hyperlink{dr}{\DR})$.
\begin{itemize}
    \item For all $i,j \in N$, we set $\alpha^S_{ij} =  D_i^S(0)$ and $\alpha^R_{ij} =  D_j^R(0)$.
    \item For all $i,j \in N$, we set $\beta^S_{it} = \frac14 D^S_i(t)$
    and $\beta^R_{jt} = \frac14 D^R_j(t)$.
\end{itemize}

Let $\DS(\alpha^S,\beta^S)$ and $\DR(\alpha^R,\beta^R)$ 
be the respective objective values of $(\hyperlink{ds}{\DS})$ and $(\hyperlink{dr}{\DR})$ for these solutions. 
For convenience, we use $\alg$ below
to denote both the algorithm and the value it outputs for the objective function.
We first show in \Cref{lemma:dual-obj} that the sum of these dual objectives upper bounds $\alg/2$. 
Then in  \Cref{lem:dual-feasible} we show that the respective dual solutions are feasible.

\begin{lemma}\label{lemma:dual-obj}
    $\DS(\alpha^S,\beta^S) + \DR(\alpha^R,\beta^R) \geq \frac12 \alg$.
\end{lemma}

\begin{proof}
  Recall that
  \begin{align*}
  \DS(\alpha^S,\beta^S) + \DR(\alpha^R,\beta^R) 
  =\sum_{i,j \in N} D_{ij} \alpha^S_{ij} - \sum_{i \in N} \sum_{t=0}^T \beta^S_{it} +
    \sum_{i,j \in N} D_{ij}\alpha^R_{ij} - \sum_{j \in N} \sum_{t=0}^T \beta^R_{jt} \ . 
  \end{align*}
 We first analyze 
    $\DS(\alpha^S,\beta^S)$.
    Consider a particular sender receiver pair $(i,j)$.
    Let $T_{ij}$ denote the collection of times 
    $t$ when $\alg $ sends a positive amount of data from $i$ to $j$. If $\alg$ does not send data from $i$ to $j$ at time $t$
    it must be the case that 
   either  $\alg$ is sending unit of data  out of sender $i$
   or into receiver $j$. 
 Since there are at most $D_i^S + D_j^R - 2D_{ij}$ such times, and sending all data from $i$ to $j$ takes at most $D_{ij}$ (fractional) time, for every $t \in T_{ij}$ it is the case that
 \begin{align*}
     t \leq D_i^S + D_j^R - D_{ij} \ .
 \end{align*}
Thus the completion time of any unit of data with
source $i$ and destination $j$ is at most  $D_i^S + D_j^R - D_{ij}$.
This gives
\begin{align*}
    \alg & \le  \sum_{i,j \in N}  D_{ij} \cdot (D_i^S + D_j^R - D_{ij}) \nonumber \\
    &\le \sum_{i,j \in N}   D_{ij} \cdot (D_i^S + D_j^R) \nonumber \\
    &= \sum_{i,j \in N} \left(  D_{ij} \alpha_{ij}^S + D_{ij} \alpha_{ij}^R \right) \ . 
\end{align*} 
The equality follows by definition of the $\alpha$ variables. Now note that 
    \begin{align*}
        \frac{1}{2} \alg &= \frac14 \sum_{t=0}^T \sum_{i \in N} D^S_{i}(t) + \frac14 \sum_{t=0}^T \sum_{j \in N} D^R_{j}(t)  
        = \sum_{t=0}^T \sum_{i \in N} \beta^S_{it} + \sum_{t=0}^T \sum_{j \in N} \beta^R_{jt} %\label{eq:dumb3}
    \end{align*}
 The first equality follows because 
    $\sum_{i} D^S_{i}(t) + \sum_{j} D^R_{j}(t)$ is equal to twice the total data at time $t$ that $\alg$ has not yet transmitted. 
    The second equality follows from the definition of $\beta$ variables. 
    Combining these two bounds on $\alg$, we conclude
\begin{align*}
    &\DS(\alpha^S,\beta^S) + \DR(\alpha^R,\beta^R) \\
    & =\sum_{i,j \in N} D_{ij} \alpha^S_{ij} - \sum_{i \in N}\sum_{t=0}^T \beta^S_{it} +
    \sum_{i,j \in N} D_{ij}\alpha^R_{ij} - \sum_{j \in N}\sum_{t=0}^T \beta^R_{jt} \\
    &\ge \alg - \frac{1}{2} \alg  =\frac{1}{2} \alg \ .
\end{align*}
This completes the proof of the lemma. 
\end{proof}

\begin{lemma}\label{lem:dual-feasible}
The dual solution    $(\alpha^S,\beta^S)$ is feasible for 
    the linear program $(\hyperlink{ds}{\DS})$ and the
    dual solution $(\alpha^R,\beta^R)$ is feasible for the linear program $(\hyperlink{dr}{\DR})$.
\end{lemma}

\begin{proof}
First consider the linear program $(\hyperlink{ds}{\DS})$.
    Since at any time at most one unit can be sent from $i$, we have
    $4 \beta^S_{it} = D^S_{i}(t) \geq D^S_i - t = \alpha_{ij}^S - t$,
    which verifies that the dual constraints of  $(\hyperlink{ds}{\DS})$ hold.
    The argumentation for the linear program $(\hyperlink{dr}{\DR})$ is symmetric. 
\end{proof}
 We now conclude the proof of \Cref{thm:main1}:
\begin{proof}[Proof of \Cref{thm:main1}]
We now conclude by noting
\begin{align*}
     4 \cdot \opt + 4 \cdot \opt 
     & \geq \opt^R + \opt^S  \\
     &\ge \DS(\alpha^S,\beta^S) + \DR(\alpha^R,\beta^R)  
      \geq \frac12 \cdot \alg \ . %\label{eq:lem2-3}
\end{align*}
The first inequality follows from the observation in Line \eqref{eq:optrs}. 
The second inequality follows from weak duality and the dual feasibility (from \Cref{lem:dual-feasible}).
Finally, the last inequality follows from  \Cref{lemma:dual-obj}.
Thus we can conclude that the approximation ratio of $\alg$ is at most $16$. 
\end{proof}

\section{Integral Matching and Indirect Routing}\label{sec:indirect-integral}

In this section, we present the proof of  \Cref{thm:main2}.
We first consider the special case where every entry of the demand matrix is $B / n$, and show worst-case upper and lower bounds on makespan and average completion time. 
At the end of this section, we will then show that
we can lift these results to the general setting while only losing small constants.

We distinguish three main cases, $B > n$, $B < 2$, and $2 \leq B \leq n$, which we handle 
in the following. All of the proposed algorithms below run in polynomial time.

\begin{lemma}
    If all demands are equal to $B / n$ and $B > n$, then there exists an indirect integral routing scheme with makespan and average completion time at most $O(B)$.
    Moreover, every indirect integral routing scheme has makespan and average completion time at least $\Omega(B)$.
\end{lemma}

\begin{proof}
    We use a {\em single round-robin} design, which is in fact a direct routing scheme. However, we will see that it achieves a best-possible makespan in the worst-case (up to a constant) even among indirect routing schemes. 
    
    The connection schedule consists of $n$ distinct matchings, equal to the set of all cyclic permutations of $n$, with each matching repeated $\ceil{\frac{B}{n}}$ times.
    With this connection schedule, every possible source-destination pair is directly connected with multiplicity $\ceil{\frac{B}{n}}$.
    The routing protocol chooses to always route flow on direct connections.
    This gives us total makespan equal to the total number of matchings, $\ceil{\frac{B}{n}}n \leq 2B$.

    To prove the lower bound, note that every indirect routing scheme must be able to route all demands out of every source, and into every destination.
    In particular, the maximum row or column sum of the demand matrix is a lower bound on the makespan, and the maximum row or column sum is exactly $B$ when all demands are equal to $B / n$.
    Moreover, the median completion time of demands transmitted by any sender is at least $B/2$, and therefore, the average completion time of all demands is $\Omega(B)$. 
\end{proof}

\begin{lemma}\label{lem:indirect-small-B}
    If all demands are equal to $B / n$ and $B < 2$, then there exists an indirect integral routing scheme with makespan and average completion time at most $O(\log n)$.
    Moreover, every indirect integral routing scheme has makespan and average completion time at least $\Omega(\log n)$.
\end{lemma}

\begin{proof}
    We use a variant of Valiant's hypercube design~\cite{vlb}.
    Assume w.l.o.g.\ that~$n$ is a power of $2$, and represent each node $u\in \{1,\hdots,n\}$ in binary. 
    That is, $u = (u_1,\hdots,u_{\log_2(n)})$ where $u_i\in \{0,1\}$ for all $i$.
    The connection schedule $\mathcal{M}$ will consist of exactly $\log_2(n)$ matchings.
    For each $i\in\{1,\hdots,\log_2(n)\}$, the $i$th matching connects nodes which match in all but the $i$th coordinate.
    That is,
    \[ M_i(u_1,\hdots,u_{\log_2(n)}) = (u_1,\hdots, u_{i-1}, u_i\oplus 1, u_{i+1}, \hdots, u_{\log_2(n)}) . \]
    The routing protocol routes flow on the shortest path from source to destination.
    That is, to route from $u$ to $v$, a single physical edge is taken for every index for which $u$ and $v$ do not match.
    This strategy obtains makespan $\log_2(n)$.
    All that is left is to show that this routing protocol does not overload the physical edges from our matchings.

    Consider a node's binary representation. 
    On average each node $u$ will agree with other nodes $v$ on half of their binary indices.
    Therefore every node pair $u,v$ is connected by a routing path with an average of $\log_2(n) / 2$ physical edges.
    Because demands are uniform, all flow in the network 
    %therefore 
    travels along $\log_2(n) / 2$ physical edges on average.
    Additionally, this also implies that the total flow $Bn$ will be evenly load-balanced across all $n \log_2 (n)$ physical edges of the network, leading to exactly 
    \[
    \frac{\log_2(n)}{2}\cdot\frac{Bn}{n \log_2(n)} \leq 1
    \] 
    flow traversing each physical edge, using $B \leq 2$.
    Thus, the flow is feasible on $\mathcal{M}$.

    To prove the lower bound, note that in order to route uniform demands $B/n$ for any $B>0$, there must be some routing path between every source-destination pair $u,v$.
    This is only possible when the number of distinct paths which leave any fixed node $u$ is at least $n$.
    Given some makespan bound $L$, the number of distinct paths leaving a fixed node $u$ 
    using exactly $1 \leq d \leq L$ physical edges is equal to~$\binom{L}{d}$.
    Thus, the number of distinct paths leaving $u$ using at most $L$ physical edges is 
    exactly $\sum_{d=1}^{L} \binom{L}{d} = 2^L-1$.
    Therefore, it must be the case that $n\leq 2^L$, or equivalently, $L\geq \log_2(n)$.
    Using the same argument on source $u$ for reaching $n/2$ destinations, we can see the median completion time
    of demands sent from $u$ is at least $\log_2(n) / 2$,
    hence the average completion time of all demands is $\Omega(\log n)$. 
\end{proof}

For the 
case
$2 \leq B \leq n$, we split our results into two lemmas. We start 
with the upper bounds.

\begin{figure}[tb]
		\centering
			\includegraphics[width=0.75\textwidth]{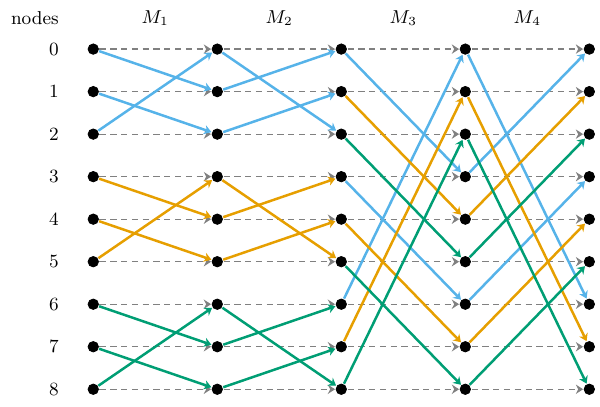}
			\caption{The time-expanded integral routing scheme of a hypercube with dimension $d=2$ and $n=9$.}\label{fig:time-exp}
\end{figure}

\begin{lemma}
    If all demands are equal to $B / n$ and $2 \leq B \leq n$, then there exists an indirect integral routing scheme with makespan and average completion time at most $O(B \log n / \log B)$.
\end{lemma}

\begin{proof}
    Let $d = \ceil{\log_2 n / \log_2 B}$ 
    and assume w.l.o.g.\ that $n^{1/d}$ is an integer.
    Represent nodes $u\in \{0,\hdots,n-1\}$ as $d$-tuples $u = (u_1,\hdots,u_d)$, where each $u_i\in\{0,\hdots,n^{1/d}-1\}$.
    For our connection schedule, we use a single period of the Elementary Basis connection schedule (Section 4.1 of~\cite{AmirWSWKA22}), illustrated in \Cref{fig:time-exp} for dimension parameter $d=2$ and explained in detail below.
    We will use multiplicity of matchings $\ceil{B /n^{1/d}}$.
    For each index $i\in\{1,\hdots,d\}$ and each scale factor $s\in \{1, \hdots, n^{1/d}-1\}$, we define the matching $M_{t=(i,s)}$ which connects nodes $u$ to nodes $v$ which match $u$ in all but the $i$th coordinates, and which differ from $u$ in the $i$th coordinate by exactly $+s( \mbox{mod } n^{1/d})$. 
    That is,
    \[ M_{t=(i,s)}(u) = (u_1,\hdots,u_{i-1},u_i + s (\mbox{mod } n^{1/d}),u_{i+1},\hdots,u_d ) . \]
    Our connection schedule repeats each matching $M_t$ exactly $\ceil{B /n^{1/d}}$ times. Hence, the length of
    the schedule is equal to $T = d \cdot (n^{1/d}-1) \cdot \ceil{B /n^{1/d}}$.
    
    To route between an arbitrary node pair $u$ to $v$, greedily route flow across any physical edge which connects to a node with decreased Hamming distance to $v$ (i.e.\ it matches $v$ in the modified $i$th coordinate).
    This process will use at most $d$ physical edges to route between each pair $u,v$.

    We next show that no edge is overloaded.
    Since there are $n^{1/d}$ choices for each of the $d$ entries
    of node's tuples, on average each node $u$ agrees with another
    node $v$ on $d / n^{1/d}$ entries. 
    Therefore, every node pair $u,v$ is connected by a routing
    path using an average of $d - d / n^{1/d}$ physical edges.
    Because uniform traffic is being routed, flow will be evenly 
    balanced across all physical edges. 
    Since a total of $Bn$ flow is routed over $Tn$ physical edges,
    the load on each edge is
    \begin{align*}
        \bigg(d - \frac{d}{n^{1/d}} \bigg) \cdot \frac{Bn}{Tn}  
        &= \bigg(d - \frac{d}{n^{1/d}} \bigg) \cdot \frac{B}{d \cdot (n^{1/d}-1) \cdot \ceil{\frac{B}{n^{1/d}}}} \\
        &= \frac{B}{n^{1/d} \cdot \ceil{\frac{B}{n^{1/d}}}} \leq 1 \ .
    \end{align*}

    Finally, we can bound the makespan as follows:
    \begin{align*}
        T &= \ceil{\frac{B}{n^{1/d}}}\cdot dn^{1/d} \\
         &\leq dn^{1/d} + dB \\
         &= \ceil{\frac{\log_2 n}{\log_2 B}} \cdot \left( n^{1/\ceil{\frac{\log_2 n}{\log_2 B}}} + B \right) \\
        & \leq \left( \frac{\log_2 n}{\log_2 B} + 1 \right) \cdot \left( n^{\frac{\log_2(B)}{\log_2(n)}} + B \right) \\
        &= \left( \frac{\log_2 n}{\log_2 B} + 1 \right) \cdot 2B \in \bigo\left( \frac{B \log_2(n)}{\log_2(B)} \right) .
    \end{align*}
    This completes the proof of the lemma. 
\end{proof}

\begin{lemma}
    If all demands are equal to $B / n$ and $2 \leq B \leq n$,
    then every indirect integral routing scheme has makespan and average completion time at least
    $\Omega(B \log n / \log B)$.
\end{lemma}

\begin{proof}
    Fix any integral routing scheme $f$ and let $T$ be its makespan.
    For any node pair $u,v$, let $S_{u,v}$ denote the minimum number
    of physical edges used to route demands from $u$ to $v$,
    and let $h$ be the median of $S_{u,v}$ over all node pairs $u,v$.
    Let $d = \frac13 \log_2 n / \log_2 B$. We distinguish whether $h \geq d$ or $h < d$.

    \textbf{Case 1:} $h \geq d$. 
    Because we assumed $f$ obtained makespan $T$, there are exactly $nT$ total physical edges that may be used to route flow.
    Additionally, there are $Bn$ total units of flow that need to be routed across the entire network. 
    In particular, each node pair $u,v$ with $S_{u,v} \geq h$ requires
    at least $h$ physical edges per unit flow. Thus,
    there is a total of $\frac12 Bn$ demand that needs at least $h$
    physical edges per unit flow.
    Thus, the average amount of flow per physical edge is 
    at least
$\frac{Bnh}{2nT} \geq \frac{Bd}{2T}$ . 
    Since each physical edge
    has unit capacity and we assumed that $f$ was feasible,
    it must be that $T \geq Bd/2 \in \Omega(B \log_2(n) / \log_2 (B))$.

    \textbf{Case 2:} $h < d$.
    By the definition of $h$, for half of the node pairs $f$ routes some
    flow using at most $h$ physical edges.
    Thus, there must exist a sender $u^*$ that can route flow
    to at least $n/2$ receivers using at most $h$ physical edges each.
    Therefore, there must exist at least $n/2$ distinct paths from $u^*$ that reach those receivers within makespan $T$ using at most $h$ physical edges.
    We next compute how many such paths exist.
    There are $\binom{T}{i}$ distinct paths using
    exactly $i$ physical edges that obtain makespan $T$. 
    Thus, there are $\sum_{i=1}^h \binom{T}{i}$ distinct paths with at 
    most $h$ physical edges within makespan $T$.
    Thus, it must be that $n/2 \leq \sum_{i=1}^h \binom{T}{i}$.

    In the following, we derive a lower bound on $T$ from this inequality.
    We first note that the proof of \Cref{lem:indirect-small-B} shows that for any $B > 0$, it holds that $T \geq \log_2 n$, and thus,
    $T \geq \log_2 n / \log_2 B = 3d$ since we assume that $B \geq 2$.
    We can now observe that
    \begin{align*}
        \sum_{i=1}^h \binom{T}{i} \leq 2 \binom{T}{h} \leq 2 \frac{T^h}{h!} \ , 
    \end{align*}
    where the first inequality holds using $h \leq d \leq T/3$, and the second inequality is a well-known approximation of the binomial coefficient.
    Rearranging $n/2 \leq 2T^h/h!$ gives us
    \[
        T \geq \bigg(h! \cdot \frac{n}{4} \bigg)^{1/h} 
        = n^{1/h} \cdot (h!)^{1/h} \cdot \bigg(\frac14 \bigg)^{1/h} \ .
    \]
    Using Stirling's approximation $h! \geq \sqrt{2\pi h}\cdot(\frac{h}{e})^h$, we can derive
    $$ T \geq \frac{h}{e}n^{1/h} \cdot \biggl( \frac{\sqrt{2\pi h}}{4} \biggr)^{1/h} \geq \frac{h}{e} n^{1/h} \ , $$
    because $(\sqrt{2\pi h}/{4})^{1/h}$ converges to $1$ from above as $h \to \infty$.
 
    The derivative of $h \cdot n^{1/h}$ is equal to $n^{1/h} (h - \ln n) / h$, which is negative
for all $2 \leq h < \ln n$. Hence, the function $h \cdot n^{1/h}$ is strictly decreasing 
on this interval.
Since $\ln n > \frac12 \log_2 n \geq \frac12 \log_2 n / \log_2 B \ge d$,
we can conclude that $h \cdot n^{1/h} \geq d \cdot n^{1/d}$ for all $2 \leq h \leq d$.
Thus, $f$ has makespan at least $T \geq h \cdot n^{1/h} / e \in \Omega(B \log_2(n) / \log_2 (B))$. This completes the proof of Case 2.

Note that by our choice of $h$, our lower bounds on the makespan also apply to the median completion time, hence also give the same lower bound on the average completion time up to constants. This completes
the proof of the lemma. 
\end{proof}

This completes our results for the uniform demand setting. To lift our
results to the general setting, we use the following well-known reduction.

\begin{lemma}
    Given a connection schedule $\mathcal{M} = M_0,\hdots,M_{T-1}$ and routing protocol $f_{unif}$ over $\mathcal{M}$ that is feasible for uniform demands $B/n$ and obtains makespan $T$,
    then for any demand matrix $D$ with row and column sums no more than $B$, we can build a routing $f_D$ over the connection schedule $\mathcal{M}\circ\mathcal{M}$ that obtains makespan $2T$.
\end{lemma}

\begin{proof}
    This is a natural application of Valiant Load Balancing~\cite{vlb}, which is described at a high level below.
    To route $r$ demand between any node pair $u,v$, first traffic is routed from $u$ to all intermediate nodes $i$ in the network in equal shares $r/n$, using $f_{unif}$.
    Then, traffic is routed from intermediate nodes to its destination using $f_{unif}$.
    This allows us to treat routing any demand $D$ with row and column sums no more than $B$ as if it is two copies of uniform demand, requesting $B/n$ demand between node pairs. 
\end{proof}

\section{Conclusion}
In this paper, we focus on
two of the variants of coflow scheduling with rational demand matrices,
for which  $O(1)$-approximations were not previously known: 
We first present a 16-approximation algorithm for minimizing the average completion time under indirect fractional matching routings. Secondly, we address both the makespan and the average completion time objectives under indirect integral routing, and provide a full characterization of
the worst-case runtimes under both objectives for this regime. 
While this falls short of providing $O(1)$-approximations for indirect integral routing, our proposed algorithm   provides $O(\log n)$-approximations for the two objectives, in addition to being worst-case optimal;  
we also believe that the worst-case analysis will be helpful 
in determining 
whether a $O(1)$-approximation for indirect integral routing is possible. 
If it is,
then we will have shown, together with the results in this paper, that constant-approximation algorithms exist for all
possible coflow scheduling variants 
in \Cref{tab:results}.

\begin{refcontext}[sorting=nyt]
\printbibliography
\end{refcontext}

\end{document}